\documentclass[journal]{IEEEtran}

\usepackage[T1]{fontenc}
\usepackage{graphicx,colordvi,psfrag}
\usepackage{amsmath,amssymb}
\usepackage{calc,pstricks, pgf, xcolor}
\usepackage{bbding}
\usepackage{bbm}
\usepackage{dsfont}
\usepackage{stfloats}

\newcommand{\Ind}{\mathds{1}}

\newcommand{\RR}{\mathbb{R}}

\newcommand{\PP}{\mathbb{P}}

\newcommand{\Ber}{\mathop{\mathrm{Ber}}}

\newcommand{\Unif}{\mathrm{Uniform}}

\def\mreals{\mathbb{R}}
\def\eqdef{\triangleq}
\def\Sph{\mathbb{S}}
\def\EE{\mathbb{E}}

\newtheorem{theorem}{Theorem}

\newtheorem{proposition}{Proposition}
\newtheorem{corollary}{Corollary}
\newtheorem{remark}{Remark}

\newtheorem{lemma}{Lemma}

\newenvironment{proof}[1][Proof]{\noindent\textbf{#1.} }{\ \rule{0.5em}{0.5em}}

\begin{document}
\title{A Note on the Probability of Rectangles for Correlated Binary Strings}

\author{Or Ordentlich, 
	Yury Polyanskiy, 
	and Ofer Shayevitz
	\thanks{}
\thanks{O. Ordentlich is with the Rachel and Selim Benin School of Computer Science and Engineering, Hebrew University of Jerusalem, Israel (email: or.ordentlich@mail.huji.ac.il).  Y. Polyanskiy is with the Massachusetts Institute of Technology, MA, USA (email: yp@mit.edu). O. Shayevitz is with the Department of Electrical Engineering Systems, Tel Aviv University, Tel Aviv, Israel (e-mail: ofersha@eng.tau.ac.il).}
\thanks{The work of O.O. was supported by the ISF under Grant 1791/17.
	The work of Y.P. was supported (in part) by the National Science Foundation under Grant No CCF-17-17842, and by the
	Center for Science of Information (CSoI), an NSF Science and Technology Center, under grant agreement CCF-09-39370. The work of O.S. was supported by the European Research Council, under grant agreement 639573.}
\thanks{}}

\date{}

\maketitle
\begin{abstract}
	Consider two sequences of $n$ independent and identically distributed fair coin tosses, $X=(X_1,\ldots,X_n)$ and
$Y=(Y_1,\ldots,Y_n)$, which are $\rho$-correlated for each $j$, i.e. $\PP[X_j=Y_j] = {1+\rho\over 2}$. We study the
question of how large (small) the probability $\PP[X \in A, Y\in B]$ can be among all sets $A,B\subset\{0,1\}^n$
of a given cardinality. For sets $|A|,|B| = \Theta(2^n)$ it is well known that the largest (smallest) probability is approximately attained by concentric (anti-concentric) Hamming balls, and this can be proved via the hypercontractive inequality (reverse hypercontractivity). Here we consider the case of $|A|,|B| = 2^{\Theta(n)}$. By applying a recent extension of the
hypercontractive inequality of Polyanskiy-Samorodnitsky (J. Functional Analysis, 2019), we show that Hamming balls of the same size
approximately maximize $\PP[X \in A, Y\in B]$ in the regime of $\rho \to 1$. We also prove a similar tight lower bound, 
i.e. show that for $\rho\to 0$ the pair of opposite Hamming balls approximately minimizes the probability $\PP[X \in A, Y\in B]$.
\end{abstract}

\section{Introduction}


Let $X\sim\Unif\left(\{0,1\}^n\right)$ and $Y\in\{0,1\}^n$ be a $\rho$-correlated copy of $X$, where $0\leq \rho<1$, i.e., 
\begin{align}
&\Pr(Y=y|X=x)\nonumber\\
& =\prod_{i=1}^n\left(\frac{1-\rho}{2}\right)^{d(x_i,y_i)}\left(\frac{1+\rho}{2}\right)^{1-d(x_i,y_i)}\nonumber\\
&=\left(\frac{1+\rho}{2}\right)^n\cdot\left(\frac{1-\rho}{1+\rho}\right)^{d(x,y)},
\end{align}
where $d(x_i,y_i)=\Ind_{\{x_i\neq y_i\}}$ and $d(x,y)=\sum_{i=1}^n d(x_i,y_i)$.
For $A,B\subset\{0,1\}^n$, we denote $P_{XY}(A\times B)\triangleq \Pr(X\in A,Y\in B)$ -- probability of a rectangle with
sides $A$ and $B$. In this paper we are interested in the following question: \textit{Among all sets of a given size, how
large/small can the probability of a rectangle be?} Previous works addressing similar questions relied on hypercontractive and reverse hypercontractive inequalities, as we describe below. Our main innovation is applying a new tool from~\cite{ps16} that is a refinement of the direct hypercontractive inequality to functions with sparse support.

A direct application of the hypercontractive inequality~\cite{EN66,Bonami1970,Beckner75,gross1975logarithmic,ODonnellBook} (see
Section~\ref{sec:hct} for more details) yields that for $A$ and $B$ of equal cardinalities, i.e. $|A|=|B| \eqdef \eta \cdot
2^n$, we have
\begin{align}
P_{XY}(A\times B)\leq \eta^{\frac{2}{1+\rho}},\label{eq:vHCTub}
\end{align}
whereas the reverse hypercontractive inequality of~\cite{CB82} was applied in~\cite{morss06} to obtain
\begin{align}
P_{XY}(A\times B)\geq \eta^{\frac{2}{1-\rho}}.\label{eq:vHCTlb}
\end{align}
Both bounds become quite tight for the regime of $\eta = \Theta(1)$, i.e. for very large sets of cardinalities
$|A|=|B|=\Theta(2^n)$. In particular,~\eqref{eq:vHCTub} is approximately attained by taking $A$ and $B$ as the
zero-centered Hamming balls containing all vectors with Hamming weight smaller than $\tfrac{n}{2}-s\sqrt{n}$, for large
$s$ independent of $n$, whereas~\eqref{eq:vHCTlb} is approximately attained by taking $A$ as such zero-centered ball and
$B$ as the same ball shifted such that its center is the all-ones vector. A special case of the construction
in~\cite{bogdanov2011extracting} also gives more constructions of sets approximately
attaining~\eqref{eq:vHCTub}: namely, 
for any $k \in \mathbb{Z_+}$ and all sufficiently large $n\ge n_0(k)$ they constructed sets $A=B$ of
 cardinality $2^{n-k}$ such that
 	\begin{equation}\label{eq:bm_2}
 		P_{XY}(A\times B) \ge \Omega_\rho(1/\sqrt{k}) 2^{-k \cdot{2\over 1+\rho}}\,,
 \end{equation}	
 thus showing that the estimate~\eqref{eq:vHCTub} is tight (up to a polylog factor $(\log {1\over \eta}))^{-\tfrac{1}{2}}$). 

In this paper we are interested in estimating the probability of rectangles for sets $A,B$ of much smaller cardinalities
(such as those frequently encountered in information and coding theories), namely $|A|=2^{n\alpha},|B|= 2^{n\beta}$ for $\alpha,\beta <
1$. Our original motivation stems from the bounds on the adder multiple access channel (MAC) zero-error capacity, obtained in~\cite{akkn18}. Sets $A,B\subset \{0,1\}^n$ are called a zero-error code for the adder MAC, if $|A+B|=|A|\cdot |B|$, where $A+B\subset\{0,1,2\}^n$ is the Minkowski sum (over the reals) of the sets $A$ and $B$. The problem of finding all pairs $(R_1,R_2)\in[0,1]^2$ for which there exist a zero-error code with sizes $|A|=2^{nR_1}$, $|B|=2^{nR_2}$ is a long standing open problem~\cite{Lindstrom69,Tilborg78,kl78,Weldon78,klwy83,ul98,os16,ap15}. One of the first results in the area, due to van Tilborg~\cite{Tilborg78}, states that if $A,B$ form a zero error code, then 
\begin{align}
W_d(A,B)&\triangleq ֿ\frac{1}{n}\log|\left\{(a,b)\in A\times B  \ : \ d(a,b)=nd\right\}|\label{eq:vt}\\
&\leq \frac{1}{n}\log{n\choose nd}+\min(d,1-d),\label{eq:vtBound}
\end{align}
for all $d\in\{0,\frac{1}{n},\ldots,1\}$. The basic idea in~\cite{akkn18} was to use~\eqref{eq:vtBound} for upper bounding 
\begin{align}
&P_{XY}(A\times B)\nonumber\\
&= 2^{-n}\left(\frac{1+\rho}{2}\right)^n\sum_{d=0}^1 2^{nW_d(A,B)} \left(\frac{1-\rho}{1+\rho}\right)^{nd}     
\end{align}
for any zero-error code $(A,B)$, and to contrast this with lower bounds on $P_{XY}(A\times B)$ for sets $|A|=2^{n R_1}$, $|B|=2^{n R_2}$ obtained in~\cite{morss06} (see Remark~\ref{remark:morss} below). A simple modification of this approach~\cite{akkn18} yielded the best known outer bounds on $(R_1,R_2)$ to date, and possibly, replacing the lower bound from~\cite{morss06} on  $P_{XY}(A\times B)$ with a sharper one, could yield stronger bounds on $(R_1,R_2)$. For instance, if our main conjecture, stated below, turns out to be true, repeating the arguments in~\cite{akkn18} with the improved bounds will yield that as $R_1$ approaches $1$ we must have that $R_2<0.4177$, improving upon $R_2<0.4228$ established in~\cite{akkn18}, which is the best known bound to date.

Our interest is in the greatest and smallest exponential decay rate of $P_{XY}(A\times B)$ among all possible sets $A,B$ of sizes $2^{n\alpha}$ and $2^{n\beta}$, respectively. To that end, for fixed $0< \alpha,\beta< 1$ we define
\begin{align}
\overline{E}(\alpha,\beta,\rho)&\triangleq- \limsup_{n\to\infty}\max_{\{A\},\{B\}} \frac{1}{n}\log P_{XY}(A\times B)\,,\\
\underline{E}(\alpha,\beta,\rho)&\triangleq-  \liminf_{n\to\infty} \min_{\{A\},\{B\}} \frac{1}{n}\log P_{XY}(A\times B)\,,
\end{align}
where $\max_{\{A\},\{B\}}$ and $\min_{\{A\},\{B\}}$ denote optimizations over the sequences of sets $A_n \subset
\{0,1\}^n$, $B_n \subset\{0,1\}^n$, $n\in \mathbb{Z}_+$ such that
	$$ |A_n| = 2^{n\alpha + o(n)}, \qquad |B_n| = 2^{n\beta + o(n)}\,.$$

Our \textbf{main conjecture} is that \textit{both $\overline{E}(\alpha,\beta,\rho)$ and $\underline{E}(\alpha,\beta,\rho)$ are optimized by concentric (resp., anti-concentric) Hamming balls}. In this work we show partial progress towards establishing this conjecture. Our conjecture is in line with the well-known facts that among all pairs of sets $A,B\subset\{0,1\}^n$ of given sizes, the maximal  distance $d_{\text{max}}(A,B)=\max_{a\in {A},b\in{B}}d(a,b)$ is minimized by concentric Hamming (quasi) balls~\cite{kleitman1966,ak77}, whereas the minimum distance $d_{\text{min}}(A,B)=\min_{a\in {A},b\in{B}}d(a,b)$ is maximized by anti-concentric Hamming (quasi) balls~\cite{ff81}.

\emph{Notation:} Logarithms are taken to base $2$ throughout, unless stated otherwise. We denote the Shannon entropy of
a random variable $V$ by $H(V)$. For a binary random variable $V\sim\Ber(p)$ we denote the entropy by $h(p)\triangleq
-p\log{p}-(1-p)\log{(1-p)}$ and its inverse restricted to $[0,1/2]$ by $h^{-1}(\cdot)$. 
For $0\leq p,q\leq 1$ we denote $p*q\triangleq p(1-q)+q(1-p)$. 

\vspace{2mm}

Our main results characterize $\overline{E}(\alpha,\alpha,\rho)$ in the low noise (large $\rho$) regime, and $\underline{E}(\alpha,\beta,\rho)$ in the high noise (small $\rho$) regime, as follows. 

\begin{theorem}\label{th:main1}
	As $\rho\to 1$ we have
	\begin{align}
	&\overline{E}(\alpha,\alpha,\rho)=(1-\alpha)\nonumber\\
	&+\frac{\frac{1}{2}-\sqrt{h^{-1}(\alpha)\left(1-h^{-1}(\alpha)\right)}}{\ln{2}}
	(1-\rho) + o(1-\rho).
	\end{align}
\label{thm:Emax}
\end{theorem}
Theorem~\ref{thm:Emax} will follow from combining Proposition~\ref{prop:EmaxUB} and Proposition~\ref{prob:EmaxLB}, proved in Section~\ref{sec:spheres} and Section~\ref{sec:hct}, respectively.


\begin{theorem}\label{th:main2}
As $\rho\to 0$ we have
\begin{align}
&\underline{E}(\alpha,\beta,\rho)=(1-\alpha)+(1-\beta)\nonumber\\
&+\rho\log{e}\left(1-2h^{-1}(\alpha)*h^{-1}(\beta)\right)+o(\rho).
\end{align}
\label{thm:Emin}
\end{theorem}

Theorem~\ref{thm:Emin} will follow from combining Proposition~\ref{prop:EminLB} and Proposition~\ref{prop:EminUB}, proved in Section~\ref{sec:spheres} and Section~\ref{sec:avdist}, respectively.

 In both cases, the optimal exponents are obtained (up to $o(\rho)$ and $o(1-\rho)$ terms) by taking $A$ and $B$ to be  Hamming spheres. In Section~\ref{sec:spheres} we compute $P_{XY}(A\times B)$ for Hamming spheres, and prove the corresponding upper bound for $\overline{E}(\alpha,\alpha,\rho)$ obtained by concentric spheres, and the lower bound on $\underline{E}(\alpha,\beta,\rho)$, obtained by spheres with opposite centers. In Section~\ref{sec:hct} we prove the lower bound on $\overline{E}(\alpha,\alpha,\rho)$.
What is interesting is that while~\eqref{eq:vHCTub} is shown via the classical hypercontractivity inequality~\cite{EN66,Bonami1970,Beckner75,gross1975logarithmic,ODonnellBook}, our result is shown by applying a recent improvement~\cite{ps16} of this inequality
for functions of small support (cf. Section~\ref{sec:hct}). In Section~\ref{sec:avdist} we prove the upper bound on
$\underline{E}(\alpha,\beta,\rho)$ by bounding the maximal average Hamming distance between members of $A$ and $B$,
subject to the cardinality constraint -- another combinatorial optimization problem of possible interest. 

\begin{remark}
After this work had been completed, we have learned from Naomi Kirshner and Alex Samorodnitsky about their concurrent
work~\cite{ks19} in which, among other things, they were able to prove that $\overline{E}(\alpha,\alpha,\rho)$ is attained by
concentric spheres for all $0<\rho<1$. Their result subsumes our Theorem~\ref{thm:Emax} and relies on a different
strengthening of a hypercontractive inequality.\footnote{In the notation of Section~\ref{sec:hct}, our work leverages the
inequality $\|T_\rho f\|_{q_0} \le \|f\|_q$ among all support-constrained functions $f$ (with the best possible $q$),
whereas the work~\cite{ks19} uses the inequality $\|T_\rho f\|_{q_0} \le e^{-n\lambda} \|f\|_{1+(q_0-1)\rho^2}$ with the
largest possible $\lambda$, which depends on the support size of $f$.} The problems of 
characterizing $\overline{E}(\alpha,\beta,\rho)$ for $\alpha \neq \beta$ and that of $\underline{E}(\alpha,\beta,\rho)$ remain open.
\end{remark}

\section{Bounds via Spheres}
\label{sec:spheres}

For $x=(x_1,\ldots,x_n)\in \{0,1\}^n$ denote the Hamming weight of
$x$ and the Hamming sphere centered at zero as
\begin{align} |x| &\eqdef |\{j: x_j=1\}|\\
	 \Sph_j &\eqdef \{x: |x|=j\} \,.
\end{align}
For the size of Hamming spheres we have~\cite[Exc. 5.8]{RG68}
\begin{equation}\label{eq:sphere}
	 |\Sph_{\lfloor d n\rfloor}| = {n \choose \lfloor d n\rfloor} = 2^{n h(d) - {1\over 2} \log n +
O(1)}\,, \qquad n\to\infty
\end{equation}
where the estimate is a consequence of Stirling's formula, $O(1)$ is uniform in $\delta$ on compact
subsets of $(0,1)$. 

Existential results (an upper bound on $\overline{E}$ and a lower bound on $\underline{E}$) follow from taking $A$ and $B$ as
Hamming spheres $\Sph_i$, $\Sph_j$ for a suitable $i,j$. Here we compute the probability of such spherical rectangles.

For any two sets $A,B\subset\{0,1\}^n$, we have
\begin{align}
P_{XY}(A\times B)&=\sum_{x\in A,y\in B}2^{-n\left(2-\log(1+\rho)-\frac{d(x,y)}{n}\log \left(\frac{1-\rho}{1+\rho} \right) \right)}\nonumber\\
&= 2^{-n(1+o(1))E(A,B,\rho)},
\end{align}
where
\begin{align}
E(A,B,\rho)\triangleq  \min_{0\leq d\leq 1}\bigg(&2-\log(1+\rho)\nonumber\\
&-W_d(A,B)-d\log \left(\frac{1-\rho}{1+\rho} \right) \bigg)\label{eq:expmin},
\end{align}
and $W_d(A,B)$ is as defined in~\eqref{eq:vt}.
Note that if $nd\notin\mathbb{N}$ we have that $W_d(A,B)=-\infty$, and therefore the minimization in~\eqref{eq:expmin} can indeed be performed on $[0,1]$ and need not be restricted to $d\in\{0,\frac{1}{n},\ldots,1\}$.

For two natural numbers $j\geq i$ and $d\in[0,1]$ such that $j-i+nd$ is even, we have that
\begin{align}
&W_d(\Sph_i,\Sph_j)\nonumber\\
&=\frac{1}{n}\log {{n}\choose{i}}{{i}\choose{\frac{1}{2}(j+i-nd)}}
{{n-i}\choose{\frac{1}{2}(j-i+nd)}}
\label{eq:Wd1}
\end{align}
for $j-i\leq nd\leq j+i$, and $W_d(\Sph_i,\Sph_j)=-\infty$ otherwise. Let $0<\alpha\leq\beta\leq 1$ and $d\in[0,1]$ be such that $i=nh^{-1}(\alpha)$ and $j=nh^{-1}(\beta)$ are integers and $j-i+nd$ is an even integer. Approximating ${n\choose k}=2^{n \left(h(k/n)+o(1)\right)}$ as in~\eqref{eq:sphere}, we have~\eqref{eq:Wder1},~\eqref{eq:Wder2} and~\eqref{eq:Wder3}  at the top of the next page,
\begin{figure*}[!t]
	\normalsize
	\begin{align}
	\frac{1}{n}\log {{n}\choose{i}}&= h(h^{-1}(\alpha))+o(1) ,\label{eq:Wder1}\\
	\frac{1}{n}\log {i\choose{\frac{1}{2}(j+i-nd)}}&= h^{-1}(\alpha) h\left(\frac{\frac{1}{2}h^{-1}(\alpha)+h^{-1}(\beta)-d}{h^{-1}(\alpha)}\right)+o(1),\label{eq:Wder2}\\
	\frac{1}{n}\log {{n-i}\choose{\frac{1}{2}(j-i+nd)}}&= \left(1-h^{-1}(\alpha)\right)h\left(\frac{\frac{1}{2}h^{-1}(\beta)-h^{-1}(\alpha)+d}{1-h^{-1}(\alpha)}\right).\label{eq:Wder3}
	\end{align}
	\hrulefill
	\vspace*{4pt}
\end{figure*}
and it therefore follows from~\eqref{eq:Wd1} that
\begin{align}
W_d(\Sph_i,\Sph_j)&=w_d(\alpha,\beta)+o(1),
\end{align}
where
\begin{align}
w_d(\alpha,\beta)&\triangleq \alpha+h^{-1}(\alpha)h\left(\frac{1}{2}+\frac{h^{-1}(\beta)-d}{2h^{-1}(\alpha)}\right)\nonumber\\
&+\left(1-h^{-1}(\alpha)\right)h\left(\frac{1}{2}+\frac{d-(1-h^{-1}(\beta))}{2(1-h^{-1}(\alpha))}\right)
\label{eq:Wd}
\end{align}
for $h^{-1}(\beta)-h^{-1}(\alpha)\leq d\leq h^{-1}(\beta)+h^{-1}(\alpha)$, and $w_d(\alpha,\beta)=-\infty$ otherwise. Since the values of $d\in[0,1]$ for which $j-i+nd$ is an even integer become arbitrarily dense as $n$ grows, by continuity of $d\mapsto w_d(\alpha,\beta)$, we have that 
\begin{align}
E(\Sph_{n h^{-1}(\alpha)},&\Sph_{n h^{-1}(\beta)},\rho) = \min_{0\leq d\leq 1}\bigg(2-\log(1+\rho)\nonumber\\
&-w_d(\alpha,\beta)-d\log \left(\frac{1-\rho}{1+\rho} \bigg) \right)+o(1)\label{eq:spheapprox}.
\end{align}

\begin{proposition}
For large $\rho$ we have
\begin{align}
&\overline{E}(\alpha,\alpha,\rho)\leq (1-\alpha)\nonumber\\
&+\frac{\frac{1}{2}-\sqrt{h^{-1}(\alpha)\left(1-h^{-1}(\alpha)\right)}}{\ln{2}}(1-\rho)+o(1-\rho).
\end{align}
\label{prop:EmaxUB}
\end{proposition}

\begin{proof}
Let $0<\alpha\leq 1$. We establish the claim by evaluating $P_{XY}(A\times B)$ for $A=B=\Sph_{n h^{-1}(\alpha)}$ and $\rho=1-\epsilon$. By~\eqref{eq:spheapprox}, it holds that
\begin{align}
&E\bigg(\Sph_{n h^{-1}(\alpha)},\Sph_{n h^{-1}(\alpha)},1-\epsilon\bigg)\nonumber\\
&=  \min_{d}\left(2-\log(2-\epsilon)-w_d(\alpha,\alpha)+d\log \left(\frac{2-\epsilon}{\epsilon} \right) \right)\nonumber\\
&+o(1)\nonumber\\
&=1+\frac{\epsilon}{2}\log(e)-\max_{d}\left(w_d(\alpha,\alpha)-d\log \left(\frac{2}{\epsilon} \right)+d \frac{\epsilon}{2}\log(e)\right)\nonumber\\
&  +o(\epsilon)+o(1).\label{eq:ExptSameSphere}
\end{align}
Denoting $r=r_{\alpha}=h^{-1}(\alpha)$, we have that 
\begin{align}
w_d(\alpha,\alpha)=h(r)+r\cdot h\left(\frac{d/2}{r}\right)+(1-r)\cdot h\left(\frac{d/2}{1-r}\right).\label{eq:samespheres}
\end{align}
The function $d\mapsto w_d(\alpha,\alpha)-d\log \left(\frac{2}{\epsilon} \right)+d \frac{\epsilon}{2}\log(e)$ is concave and its derivative
\begin{align}
\frac{1}{2}\log\left(\frac{1-\frac{d/2}{r}}{\frac{d/2}{r}}\right)&+\frac{1}{2}\log\left(\frac{1-\frac{d/2}{1-r}}{\frac{d/2}{1-r}}\right)\nonumber\\
&-\log\left(\frac{2}{\epsilon}\right)+ \frac{\epsilon}{2}\log(e)
\end{align}
equals zero at $d^*=\epsilon\sqrt{r(1-r)}+o(\epsilon)$. Thus, the optimizing $d$ in~\eqref{eq:ExptSameSphere} is $d^*=\epsilon\sqrt{r(1-r)}+o(\epsilon)$, and therefore
\begin{align}
&E\bigg(\Sph_{n h^{-1}(\alpha)},\Sph_{n h^{-1}(\alpha)},1-\epsilon\bigg)=1-h(r)\nonumber\\
&+\frac{\epsilon}{2}\log(e)+\epsilon\sqrt{r(1-r)}+\epsilon\log\left(\frac{1}{\epsilon}\right)\sqrt{r(1-r)}\nonumber\\
&-\left[r\cdot h\left(\sqrt{\frac{1-r}{r}}\frac{\epsilon}{2}\right)+(1-r)\cdot h\left(\sqrt{\frac{r}{1-r}}\frac{\epsilon}{2}\right)\right]\nonumber\\
&+o(\epsilon)+o(1)\label{eq:longEnt}
\end{align}
We approximate the term in the square brackets in equations~\eqref{eq:SQder1},~\eqref{eq:SQder2} and~\eqref{eq:longEntSum} at the bottom of the page.
\begin{figure*}[!b]
	\normalsize
		\hrulefill
	\begin{align}
	r\cdot h\left(\sqrt{\frac{1-r}{r}}\frac{\epsilon}{2}\right)&=-\sqrt{r(1-r)}\frac{\epsilon}{2}\log\left(\sqrt{\frac{1-r}{r}}\frac{\epsilon}{2}\right)-r\left(1-\sqrt{\frac{1-r}{r}}\frac{\epsilon}{2}\right)\log\left(1-\sqrt{\frac{1-r}{r}}\frac{\epsilon}{2}\right)\nonumber\\
	&=-\sqrt{r(1-r)}\frac{\epsilon}{2}\log\left(\sqrt{\frac{1-r}{r}}\frac{\epsilon}{2}\right)+\frac{\epsilon}{2}\sqrt{r(1-r)}\log{(e)}+o(\epsilon), ,\label{eq:SQder1}\\
	(1-r)\cdot h\left(\sqrt{\frac{r}{1-r}}\frac{\epsilon}{2}\right)&=-\sqrt{r(1-r)}\frac{\epsilon}{2}\log\left(\sqrt{\frac{r}{1-r}}\frac{\epsilon}{2}\right)-(1-r)\left(1-\sqrt{\frac{r}{1-r}}\frac{\epsilon}{2}\right)\log\left(1-\sqrt{\frac{r}{1-r}}\frac{\epsilon}{2}\right)\nonumber\\
	&=-\sqrt{r(1-r)}\frac{\epsilon}{2}\log\left(\sqrt{\frac{r}{1-r}}\frac{\epsilon}{2}\right)+\frac{\epsilon}{2}\sqrt{r(1-r)}\log{(e)}+o(\epsilon),\label{eq:SQder2}\\
	r\cdot h\left(\sqrt{\frac{1-r}{r}}\frac{\epsilon}{2}\right)&+(1-r)\cdot h\left(\sqrt{\frac{r}{1-r}}\frac{\epsilon}{2}\right)=-\sqrt{r(1-r)}\epsilon\log\left(\frac{\epsilon}{2}\right)+\epsilon\sqrt{r(1-r)}\log{(e)}+o(\epsilon)\nonumber\\
	&=\sqrt{r(1-r)}\epsilon\log\left(\frac{1}{\epsilon}\right)+\epsilon\sqrt{r(1-r)}+\epsilon\sqrt{r(1-r)}\log{(e)}+o(\epsilon)
	.\label{eq:longEntSum}
	\end{align}
	\vspace*{4pt}
\end{figure*}
Substituting~\eqref{eq:longEntSum} into~\eqref{eq:longEnt}, we obtain
\begin{align}
&E\bigg(\Sph_{n h^{-1}(\alpha)},\Sph_{n h^{-1}(\alpha)},1-\epsilon\bigg)=1-h(r)\nonumber\\
&+\left(\frac{1}{2}-\sqrt{r(1-r)}\right)\epsilon\log{(e)}+o(\epsilon)+o(1).
\end{align}
The claim now follows by definition of $\overline{E}(\alpha,\alpha,\rho)$.
\end{proof}

\begin{proposition}
For small $\rho$ we have that
\begin{align}
\underline{E}(\alpha,\beta,\rho)&\geq (1-\alpha)+(1-\beta)\nonumber\\
&+\rho\log{e}\left(1-2h^{-1}(\alpha)*h^{-1}(\beta)\right)+o(\rho).
\end{align}
\label{prop:EminLB}
\end{proposition}

\begin{proof}
We establish the claim by evaluating $P_{XY}(A\times B)$ for $A=\Sph_{n h^{-1}(\alpha)}$ and $B=1^n+\Sph_{n h^{-1}(\beta)}$, i.e., a zero-centered Hamming sphere and a Hamming sphere centered around the all-ones vector $1^n$. First, note that for any $A,B\subset\{0,1\}^n$ it holds that
\begin{align}
W_d(A,1^n+B)=W_{1-d}(A,B).
\end{align}
Thus, applying~\eqref{eq:spheapprox}, we see that for $0<\alpha\leq\beta\leq 1$ it holds that
\begin{align}
&E\left(\Sph_{n h^{-1}(\alpha)},1^n+\Sph_{n h^{-1}(\beta)},\rho\right)\nonumber\\
&=\min_{0\leq d\leq 1}\bigg(2-\log(1+\rho)-w_d(\alpha,\beta)\nonumber\\
&~~~~~~~~~~~~~-(1-d)\log \left(\frac{1-\rho}{1+\rho} \right) \bigg)+o(1)\nonumber\\
&=2-\log(1-\rho)-\max_{0\leq d\leq 1}\left(w_d(\alpha,\beta)-d\log \left(\frac{1-\rho}{1+\rho} \right) \right)\nonumber\\
&+o(1)\label{eq:expmin2}
\end{align}
Let us consider the case of $\rho\ll 1$. In this case, we have that $\log(1+\rho)=\rho\log{e}+o(\rho)$, so that~\eqref{eq:expmin2} reads
\begin{align}
&E\left(\Sph_{n h^{-1}(\alpha)},1^n+\Sph_{n h^{-1}(\beta)},\rho\right)=2+\rho\log{e}\nonumber\\
&- \max_{d}\left(w_d(\alpha,\beta)+2d\rho\log{e}\right)+o(\rho)+o(1).\label{eq:expmin3}
\end{align} 
The function $d\mapsto w_d(\alpha,\beta)\triangleq g(d)$ is strictly concave, and it is straightforward to verify that $g'(h^{-1}(\alpha)*h^{-1}(\beta))=0$ and that $g(h^{-1}(\alpha)*h^{-1}(\beta))=\alpha+\beta$. Denoting $c=2g''(h^{-1}(\alpha)*h^{-1}(\beta))<0$ and setting $\delta=d-h^{-1}(\alpha)*h^{-1}(\beta)$, we therefore have
\begin{align}
g(d)=\alpha+\beta+c \delta^2+o(\delta^2).
\end{align}
Consequently,
\begin{align}
&w_d(\alpha,\beta)+2d\rho\log{e}=g(d)+2d\rho\log{e}\nonumber\\
&=\alpha+\beta+c \delta^2+2(h^{-1}(\alpha)*h^{-1}(\beta)+\delta)\rho\log{e}+o(\delta^2)\nonumber\\
&=\alpha+\beta+\rho\log{e}\cdot 2h^{-1}(\alpha)*h^{-1}(\beta)\nonumber\\
&+\delta\left(2\rho\log{e}+c\delta+o(\delta)\right)\nonumber\\
&\leq \alpha+\beta+\rho\log{e}\cdot 2h^{-1}(\alpha)*h^{-1}(\beta)+o(\rho),\label{eq:wdub}
\end{align}
where the last inequality follows since $c<0$. Substituting~\eqref{eq:wdub} into~\eqref{eq:expmin3} we obtain
\begin{align}
&E\left(\Sph_{n h^{-1}(\alpha)},1^n+\Sph_{n h^{-1}(\beta),\beta}\right)\geq (1-\alpha)+(1-\beta)\nonumber\\
&+\rho\log{e}\left(1-2h^{-1}(\alpha)*h^{-1}(\beta)\right)+o(\rho)+o(1).
\end{align} 
The claim now follows by definition of $\underline{E}(\alpha,\beta,\rho)$.
\end{proof}

\section{Lower Bound on $\overline{E}(\alpha,\alpha,\rho)$}
\label{sec:hct}

For a function $f:\{0,1\}^n\to\RR^+$ and $p\geq 1$ we define $\|f\|_p=\mathbb{E}^{1/p}[|f(X)|^p]$. For a set $A\subset \{0,1\}^n$
denote
	$$\Ind_A(x) \eqdef \begin{cases} 0, & x\not\in A\\
					 1, & x \in A
			\end{cases} $$
We have that
\begin{align}
P_{XY}(A\times B)&=\mathbb{E}\left[\Ind_A(X)\Ind_B(Y)\right]\nonumber\\
&=\mathbb{E}\left[\Ind_B(Y)\mathbb{E}\left[\Ind_A(X)|Y\right]\right]\nonumber\\
&=\mathbb{E}\left[\Ind_B(Y)(T_{\rho}\Ind_A)(Y)\right]\,,
\end{align}
where
\begin{align}
(T_\rho f)(y)\triangleq\mathbb{E}[f(X)|Y=y].
\end{align}
Denoting the inner-product $(f,g) = \EE[f(Y) g(Y)]$ and noticing that $T_\rho$ is self-adjoint and satisfies
the semigroup property $T_{\rho_1}T_{\rho_2}=T_{\rho_1\rho_2}$ (for $0<\rho_1,\rho_2<1$), we obtain
\begin{align} P_{XY}(A\times B)&= (\Ind_B, T_\rho \Ind_A)\nonumber\\
		   &=(T_{\rho_1} \Ind_B, T_{\rho_2} \Ind_A) \qquad \forall \rho_1\rho_2 = \rho\label{eq:iterexpt}\\
		   &\le \|T_{\rho_1} \Ind_B\|_{2} \|T_{\rho_2} \Ind_A\|_{2}\label{eq:holder}\,,
\end{align}
where the last step is Cauchy-Schwarz inequality.

The next step is to use the hypercontractivity inequality to upper bound $\|T_\rho f\|_p$. Denote the support size of
$f$ by $\|f\|_0$. Since $\|f\|_0\ll 2^n$, we will use an improved hypercontractivity inequality from~\cite{ps16}, that takes $\|f\|_0$
into account. The following result is a key ingredient:
\begin{theorem}[Theorem 7 in \cite{ps16}]\label{th:hcb} Fix $1<p_0 < \infty$ and $0\le \lambda_0 \le(1-p_0^{-1})\ln2$. 
For any $f:\{0,1\}^n \to \mreals_+$ with $\|f\|_{p_0} \ge e^{n\lambda_0}
\|f\|_1$ we have
\begin{equation}\label{eq:hcb_1}
	\|T_{e^{-t}} f\|_{p(t)} \le \|f\|_{p_0}\,, \qquad p(t) = 1+e^{u(t)}\,,
\end{equation}
where $u(t)$ is the unique solution on $[0,\infty)$ of the following ODE with initial condition $u(0)=\ln(p_0-1)$
\begin{subequations}\label{eq:hcb_0} 
\begin{align}
	u'(t) &= C\left(\lambda_0(1+e^{-u(t)})\right)\\ 
	C(\ln 2(1 - h(y))) &= {2-4\sqrt{y(1-y)} \over \ln 2(1
- h(y))}.
\end{align}
\end{subequations}
Furthermore, the function $C:[0,\ln2]\to[2, 2/\ln2]$ is a smooth, convex and strictly increasing bijection.
\end{theorem}

From this result we derive the following implication for indicator functions.
\begin{theorem}\label{th:hcbx} Fix $0<\alpha<1$ and $1 < q_0 < \infty$.  Then there exists a function $q=q(t)$ defined on an interval
$t\in[0,\epsilon)$ for some $\epsilon>0$ such that for all sets $A\subset \{0,1\}^n$ with $|A|\le 2^{n\alpha}$ we have
\begin{equation}\label{eq:hx_a}
		\|T_{e^{-t}} 1_A\|_{q_0} \le \|1_A\|_{q(t)} \qquad \forall t \in [0,\epsilon)\,.
\end{equation}	
The function $q(t)$ satisfies  
\begin{equation}\label{eq:hx_b}
		q(t) = q_0 - (q_0-1) C((1-\alpha) \ln 2)t + O(t^2) \qquad \mbox{as~} t\to 0\,.
\end{equation}	
\end{theorem}
\begin{remark} Note that the standard hypercontractivity estimate~\cite{EN66,Bonami1970,Beckner75,gross1975logarithmic} yields the same result without
restriction on the size of the set $A$ but with a strictly worse (larger) function $q(t) = (q_0-1)e^{-2t}+1$. See~\cite[Remark 3]{ps16}.
\end{remark}
\begin{proof}
	Denote by $u_f(a, b, t)$ the solution of the ordinary differential equation (ODE)
	$$ {d\over dt} u(t) = C(b(1+e^{-u(t)}))\,,$$
	with $u(0) = a$. Here $C(\cdot)$ is a function defined in~\eqref{eq:hcb_0}, $a \in \mathbb{R}$ and $0< b
	< (1+e^{-a})^{-1}\ln 2$. For a fixed $a,b$ the standard results on ODEs imply that
	this solution exists and is unique in some neighborhood $-\epsilon < t < \epsilon$ of zero. Furthermore, for any
	$a_0,b_0$ satisfying $0<b_0<(1+e^{-a_0})^{-1}\ln 2$ there exists an $\epsilon_1>0$ such that the map
		$$ (a,b,t) \mapsto u_f(a,b,t) $$
	is smooth for $|a-a_0|<\epsilon_1, |b-b_0|<\epsilon_1, |t| < \epsilon_1$ (for both of these results,
	cf.~\cite[Chapter 2, Section 7, Corollary 6]{ArnoldODE}. We set $a_0 = \ln(q_0-1)$ and $b_0 =
	(1-\alpha)(1-q_0^{-1})\ln 2$. We will call triplets $(a,b,t)$ in the above neighborhood of
	$(a_0,b_0,0)$  \textit{admissible}.

	From~\eqref{eq:hcb_1} we have for any admissible $(a,b,s)$ with $s\ge 0$ and any $A$ with $|A|\le 2^{n\alpha}$:
	\begin{equation}\label{eq:hx_2}
			\|T_{e^{-s}} 1_A\|_{1+e^{u_f(a, b, s)}} \le \|1_A\|_{1+e^{a}}\,,
	\end{equation}
	provided that $b(1+e^{-a}) \le (1-\alpha)\ln 2$ (this is just the condition $\|f\|_{p_0} \ge e^{n\lambda_0}
	\|f\|_1$ of Theorem~\ref{th:hcb}).
	
	Our aim is to set $s=t$ in~\eqref{eq:hx_2} and show that there exists a choice of $a = a(t)$ and $b =
	b(t)$ and $\epsilon < \epsilon_1$ such that the following conditions are satisfied: (C1) $a(0)=a_0$, $b(0)=b_0$ and both functions are
	smooth on $|t| <\epsilon$; (C2) for any $|t|<\epsilon$ the triplet $(a(t),b(t),t)$ is admissible; 
	(C3) for each $|t|<\epsilon$ 
		\begin{equation}\label{eq:hx_9}
			\left\{ \begin{array}{ll} 
			b(t)(1+e^{-a(t)}) &= (1-\alpha)\ln 2\\
					 u_f(a(t),b(t),t) &= \ln(q_0-1)
				\end{array} \right. 
\end{equation}				
	It is clear that if indeed such a choice of $a(t),b(t)$ were found we get from~\eqref{eq:hx_2} with $s=t$ the
	statement of the Theorem with $q(t) =1+e^{a(t)}$.

	We claim that it is sufficient to show that the system of equations
	\begin{align}
	 \left\{ \begin{array}{ll} f(a,b) &= 0,\\
					 u_f(a,b,t) &= \ln(q_0-1)
				\end{array} \right. 
	\end{align}
	where $ f(a,b) \eqdef b-(1-\alpha) (1-(1+e^a)^{-1}) \ln 2$, is uniquely solvable (for $a,b$) in the interval $-\epsilon < t < \epsilon$ and that solution $a(t),b(t)$ is
	smooth. Indeed, since the triplet $(a_0,b_0,0)$ is a solution, we get (C1). Smoothness of $a(t),b(t)$ implies
	(C2). And, finally, (C3) is automatic. Smooth solvability, in turn, follows from the fact that the map
		\begin{equation}\label{eq:hx_6}
			(a,b,t) \mapsto (f(a,b), u_f(a,b,t), t) 
		\end{equation}
	has non-trivial Jacobian at $(a_0,b_0,0)$. Indeed, denoting $\partial_x = {\partial \over \partial x}$ 
	the Jacobian is given by
		$$ \mathrm{Jac}(a,b,t) = (\partial_a f) (\partial_b u_f) - (\partial_b f) (\partial_a u_f)\,. $$
	To evaluate this we note an identity $ u_f(a, b, 0) = a$ and thus
		\begin{align} \left.{\partial \over \partial a}\right|_{t=0} u_f(a,b,t) & = 1,\label{eq:hx_4}\\
		   \left.{\partial \over \partial b}\right|_{t=0} u_f(a,b,t) & = 0, \label{eq:hx_5}\\
		   \left.{\partial \over \partial t}\right|_{t=0} u_f(a,b,t) & = C(b (1+e^{-a})). \label{eq:hx_7}
		\end{align}
	Therefore, at $(a=a_0,b=b_0,t=0)$ the Jacobian evaluates to 
		$$ \mathrm{Jac}(a_0,b_0,0) =  -1 \neq 0\,. $$
	Since the Jacobian is non-zero in some neighborhood of $(a_0,b_0,0)$, the map~\eqref{eq:hx_6} can be
	locally inverted, and we take for $a(t), b(t)$ the pre-image of $(0,0,t)$ under~\eqref{eq:hx_6}.

	Finally, we need to show that $q(t) = 1+e^{a(t)}$ satisfies the expansion~\eqref{eq:hx_b}. To that end, we
	differentiate over $t$ the identity
	\begin{align}
	u_f(a(t), b(t), t) = \ln(q_0-1)
	\end{align}  
	to get
	\begin{align} 
	\dot{a}(t) {\partial_a} u_f(a(t),b(t), t) &+ \dot{b}(t) \partial_b u_f(a(t),b(t),t) \nonumber\\
	&+ \partial_t
		u_f(a(t),b(t), t) = 0
		\end{align}
		where $\dot{a}(t) \eqdef {da(t)\over dt}$ and $\dot{b}(t) \eqdef {db(t)\over dt}$.
	At $t=0$ this is evaluated via~\eqref{eq:hx_4}-\eqref{eq:hx_7} to give
	\begin{align}
	\dot{a}(0) + C((1-\alpha)\ln 2) = 0\,.
	\end{align}
	This clearly implies that $q(t) =  1+e^{a(t)}$ satisfies~\eqref{eq:hx_b}.
\end{proof}

The following application of the previous result establishes the hard direction of Theorem~\ref{th:main1}.

\begin{proposition} Fix $\rho\in(0,1)$. Then for any sets 
$A,B$ with $|A|\le 2^{n\alpha}$, $|B|\le 2^{n\alpha}$ we have
\begin{equation}\label{eq:pp_1}
	P_{XY}(A\times B)\leq 2^{-n \psi(\alpha,\rho) }\,,
\end{equation}
where as $\rho\to 1$ we have
	\begin{align}\label{eq:pp_2}
		&\psi(\alpha,\rho) = (1-\alpha) \nonumber\\
		&+ {1\over \ln 2} (1/2 - \sqrt{h^{-1}(\alpha) (1-h^{-1}(\alpha))})
	(1-\rho) \nonumber\\
	&+ o(1-\rho) \,.
\end{align}
\label{prob:EmaxLB}
\end{proposition}
\begin{remark} For bounding $\overline{E}(\alpha,\beta,\rho)$ with $\alpha\neq \beta$ this method does not give a bound
matching that attained by Hamming spheres. The main reason is that if we take $A,B$ as concentric (but grossly unequal)
Hamming balls the Cauchy-Schwarz inequality~\eqref{eq:holder} is applied to functions $T_{\rho_1}\Ind_A$, $T_{\rho_2}\Ind_B$
which have effectively disjoint supports for $\rho\to1$.
\end{remark}
\begin{proof} Let $\rho = e^{-2t}$ for some fixed $t$. Suppose the sets $A,B$ both have sizes at most
$2^{n\alpha}$. Then from Theorem~\ref{th:hcbx}
we obtain
\begin{subequations}\label{eq:hcb_0} 
	\begin{align} 
	& \|T_{e^{-t}} \Ind_A\|_{2} \le \|\Ind_A\|_{p(t)}\\
	 &\|T_{e^{-t}} \Ind_B\|_{2} \le \|\Ind_B\|_{p(t)}\\ 
&p(t) = 2 - (2-1)C((1-\alpha)\ln
	2)t + o(t).
\end{align}	
\end{subequations}
Since $\|\Ind_A\|_q = 2^{-n(1-\alpha)/q}$ we get from~\eqref{eq:holder} the following:
	\begin{align} 
	&{1\over n} \log P_{X,Y}(A\times B) \le -\frac{2}{p(t)}(1-\alpha)\\
	&=-(1-\alpha)\left(1+\frac{t}{2}C((1-\alpha)\ln
	2)+o(t)\right) \\
	&= -(1-\alpha)-\frac{t(1-\alpha)}{2}\frac{2-4\sqrt{h^{-1}(\alpha)(1-h^{-1}(\alpha))}}{(1-\alpha)\ln{2}}\nonumber\\
	&+o(t).
\end{align}
	The statement now follows since $t=\frac{1-\rho}{2}+o(1-\rho)$. 
\end{proof}

\section{Upper Bound on $\underline{E}(\alpha,\beta,\rho)$}
\label{sec:avdist}

Note that
\begin{align}
&P_{XY}(A\times B)=\sum_{a\in A,b\in B} \Pr(X=a,Y=b)\nonumber\\
&=|A|\cdot|B|\nonumber\\
&\cdot\frac{1}{|A|\cdot|B|}\sum_{a\in A,b\in B} 2^{-n}\left(\frac{1+\rho}{2}\right)^n\cdot\left(\frac{1-\rho}{1+\rho}\right)^{d(a,b)}\nonumber\\
&\geq  |A|\cdot|B|\nonumber\\
&\cdot 2^{-n}\left(\frac{1+\rho}{2}\right)^n\cdot\left(\frac{1-\rho}{1+\rho}\right)^{\frac{1}{|A|\cdot|B|}\sum_{a\in A,b\in B}d(a,b)}\label{eq:jensen}\\
&=2^{-n\left(2-\frac{\log{(|A|\cdot|B|)}}{n}-\log(1+\rho)-\frac{\log\frac{1-\rho}{1+\rho}}{|A|\cdot|B|}\sum_{a\in A,b\in B}\frac{d(a,b)}{n}\right)},\label{eq:Plb}
\end{align}
where we have used Jensen's inequality in~\eqref{eq:jensen}. As $\frac{1-\rho}{1+\rho}<1$, we need to upper bound $\frac{1}{|A|\cdot|B|}\sum_{a\in A,b\in B}d(a,b)$ in terms of $|A|$ and $|B|$ in order to further lower bound~\eqref{eq:Plb}. Consequently, we define
\begin{align}
&\bar{d}(n,\alpha,\beta)\nonumber\\
&=\frac{1}{n}\max_{A,B: |A|=2^{n\alpha},|B|=2^{n\beta}} \frac{1}{|A|\cdot|B|}\sum_{a\in A,b\in B}d(a,b)\label{eq:minAvgDist}\\
&\underline{d}(n,\alpha,\beta)\nonumber\\
&=\frac{1}{n}\min_{A,B: |A|=2^{n\alpha},|B|=2^{n\beta}} \frac{1}{|A|\cdot|B|}\sum_{a\in A,b\in B}d(a,b).
\end{align}
With these definitions we relax~\eqref{eq:Plb} to
\begin{align}
&P_{XY}(A\times B) \nonumber\\
&\geq 2^{-n\left((1-\alpha)+(1-\beta)-\log(1+\rho)-\bar{d}(n,\alpha,\beta)\log\frac{1-\rho}{1+\rho}\right)}.\label{eq:Pnlb}
\end{align}
It is obvious that $\bar{d}(n,\alpha,\beta)=1-\underline{d}(n,\alpha,\beta)$, since if the sets $(A,B)$ achieve the minimal average distance, the sets $(A,B'=1^n+B)$ must achieve the maximal average distance. A quantity similar to $\underline{d}(n,\alpha,\beta)$, where the optimization in~\eqref{eq:minAvgDist} is performed over all families $A$ of size $2^{n\alpha}$ while $B=A$ was defined in~\cite[p.10 eq. 1]{ak77}, and its asymptotic (in $n$) value, was characterized in~\cite{aa94}. Below we prove a lower bound on $\underline{d}(n,\alpha,\beta)$. The technique is quite similar to that of~\cite{aa94}, and requires the following simple proposition.

\begin{proposition}
	The function $\varphi(x,y)=h^{-1}(x)*h^{-1}(y)$ is jointly convex in $(x,y)\in[0,1]^2$.
	\label{prop:convexity}
\end{proposition}

The function $\varphi(x,y)$ is plotted in Figure~\ref{fig:conv}. To prove Proposition~\ref{prop:convexity}, we will rely on the following simpler statement, which is essentially proved in~\cite{aa94}. For completeness we provide the proof in the appendix.

\begin{proposition}
	The function $x\mapsto h^{-1}(x)\left(1-h^{-1}(x)\right)$ is convex in $[0,1]$.
	\label{prop:convexityScalar}
\end{proposition}

\begin{proof}[Proof of Proposition~\ref{prop:convexity}]
Let $(X,Y)$ be two (possibly dependent) random variables on $[0,1]^2$. We use the identity $a*b=\frac{1}{2}(1-(1-2a)(1-2b))$ to write
\begin{align}
&\mathbb{E}[\varphi(X,Y)]\nonumber\\
&=\frac{1}{2}\left(1-\mathbb{E}\left[\left(1-2h^{-1}(X)\right)\left(1-2h^{-1}(Y)\right)\right]\right)\nonumber\\
&\geq \frac{1}{2}\bigg(1-\sqrt{\mathbb{E}\left[\left(1-2h^{-1}(X)\right)^2\right]}\\
&~~~~~~~~~~~~~~\sqrt{\mathbb{E}\left[\left(1-2h^{-1}(Y)\right)^2\right]}\bigg)\label{eq:CSineq}\\
&\geq \frac{1}{2}\bigg(1-\sqrt{\left(1-2h^{-1}(\mathbb{E}\left[X\right])\right)^2}\nonumber\\
&~~~~~~~~~~~~~~\sqrt{\left(1-2h^{-1}(\mathbb{E}\left[Y\right])\right)^2}\bigg)\label{eq:scalarconvexity}\\
&=\varphi(\mathbb{E}[X],\mathbb{E}[Y]),
\end{align}
where~\eqref{eq:CSineq} follows from the Cauchy-Schwarz inequality, and~\eqref{eq:scalarconvexity} from Jensen's inequality and the fact that $t\mapsto(1-2h^{-1}(t))^2=1-4h^{-1}(t)\left(1-h^{-1}(t)\right)$ is concave due to Proposition~\ref{prop:convexityScalar}.
\end{proof}

\begin{figure*}[ht]
	\centering
	\includegraphics[width=1.8\columnwidth]{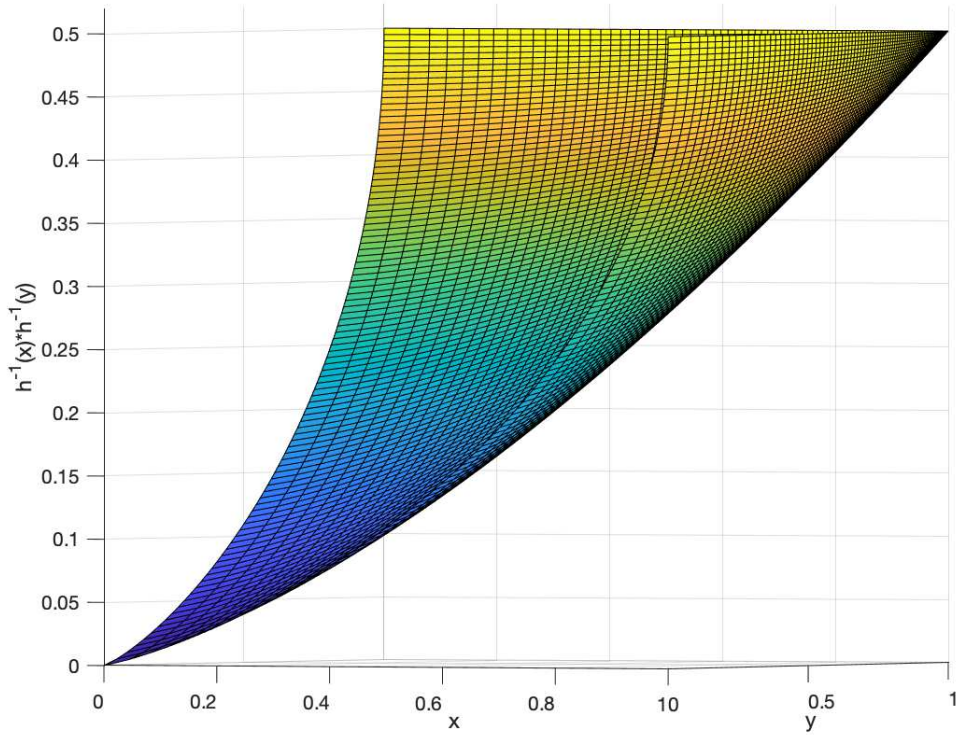}
	\caption{\label{fig:conv}  Illustration of the function $h^{-1}(x)*h^{-1}(y)$.}	
\end{figure*}

\begin{lemma}
	For any two independent $n$-dimensional random binary vectors $V$ and $W$  
	\begin{align}
	h^{-1}\left(\frac{H(V)}{n}\right)&*h^{-1}\left(\frac{H(W)}{n}\right)\leq\frac{\mathbb{E}d(V,W)}{n}\nonumber\\
	&\leq 1-h^{-1}\left(\frac{H(V)}{n}\right)*h^{-1}\left(\frac{H(W)}{n}\right).
	\end{align}
	\label{lem:dmaxim}
\end{lemma}

\begin{proof}
	Let $V$ and $W$ be two independent random vectors with $H(V)=n\alpha$ and $H(W)=n\beta$. Further, let $a_i\triangleq \Pr(V_i=1)$, $b_i\triangleq \Pr(W_i=1)$, be the induced marginal distributions for each coordinate.
	Our goal is to minimize and maximize $\sum_{i=1}^n a_i*b_i$ under the entropy constraints $H(V)=n\alpha,H(W)=n\beta$. We may and will assume without loss of generality that $a_i,b_i\leq 1/2$ for all $i$.
	We have
	\begin{align}
	\inf_{\substack{{V,W:}\\{H(V)=n\alpha}\\{H(W)=n\beta}}}&\sum_{i=1}^n a_i*b_i\geq \inf_{\substack{{V,W:}\\{H(V)\geq n\alpha}\\{H(W)\geq n\beta}}}\sum_{i=1}^n a_i*b_i\nonumber\\
	&=\inf_{\substack{{\{a_i\},\{b_i\}:}\\{\sum_{i=1}^n h(a_i)\geq n\alpha}\\{\sum_{i=1}^n h(b_i)\geq n\beta}}}\sum_{i=1}^n a_i*b_i\label{eq:marginalization}\\
	&=\inf_{\substack{{\{\alpha_i\},\{\beta_i\}:}\\{\tfrac{1}{n}\sum_{i=1}^n \alpha_i \geq  \alpha}\\{\frac{1}{n}\sum_{i=1}^n \beta_i \geq \beta}}}\sum_{i=1}^n h^{-1}(\alpha_i)*h^{-1}(\beta_i)
	\end{align}
	where~\eqref{eq:marginalization} follows since the cost function $\sum_{i=1}^n a_i*b_i$ depends only on the marginal distributions, and for every feasible distribution $V,W$ the product of the marginalized distributions is also feasible. Our lower bound now immediately follows from Proposition~\ref{prop:convexity}. For the upper bound, note that if $V$ and $W$ minimize $\mathbb{E}d(V,W)$ under the entropy constraints, $V$ and $W'=W+1^n$ maximizes the expected distance under the same entropy constraints.
\end{proof}

Taking $V\sim\Unif(A)$ and $W\sim\Unif(B)$, we immediately get the following.
\begin{corollary}
\begin{align}
\bar{d}(n,\alpha,\beta)&\leq n\left(1-h^{-1}\left(\alpha\right)*h^{-1}\left(\beta\right)\right),\\
\underline{d}(n,\alpha,\beta)&\geq nh^{-1}\left(\alpha\right)*h^{-1}\left(\beta\right).
\end{align}
\label{cor:d}
\end{corollary}

Combining~\eqref{eq:Pnlb} and Corollary~\ref{cor:d}, gives
\begin{align}
\underline{E}(\alpha,\beta,\rho)&\leq (1-\alpha)+(1-\beta)-\log(1+\rho)\nonumber\\
&-\left(1-h^{-1}\left(\alpha\right)*h^{-1}\left(\beta\right)\right)\log\frac{1-\rho}{1+\rho}\nonumber\\
&=(1-\alpha)+(1-\beta)-\log(1-\rho)\nonumber\\
&+\left(h^{-1}\left(\alpha\right)*h^{-1}\left(\beta\right)\right)\log\frac{1-\rho}{1+\rho}.\label{eq:trivial}
\end{align}
We have therefore obtained the following.

\begin{proposition}
We have
\begin{align}
\underline{E}(\alpha,\beta,\rho)&\leq (1-\alpha)+(1-\beta)\nonumber\\
&+\rho\log{(e)}\left(1-2h^{-1}\left(\alpha\right)*h^{-1}\left(\beta\right)\right)+o(\rho).\label{eq:underEub}
\end{align}
\label{prop:EminUB}
\end{proposition}

\begin{remark}
In~\cite{morss06} the bound
\begin{align}
\underline{E}(\alpha,\beta,\rho)\leq \frac{(1-\alpha)+(1-\beta)+2\rho\sqrt{(1-\alpha)(1-\beta)}}{1-\rho^2}\label{eq:MORSS}
\end{align}
was proved, using reverse hypercontractivity. It is easy to verify that for $\alpha=\beta$ the bound~\eqref{eq:trivial} is strictly better than~\eqref{eq:MORSS} for all $\alpha<1-\tfrac{1-\rho}{2\rho}\log\left(\frac{1}{1-\rho}\right)$. Moreover, for any $0<\alpha,\beta<1$ the bound~\eqref{eq:trivial} is better than~\eqref{eq:MORSS} for $\rho$ large enough. The reverse hypercontractivity bound states that  for $p<1$ we have $\|T_{\rho}f\|_{q(\rho,p)}\geq \|f\|_p$ where $q(\rho,p)=1-\frac{1-p}{\rho^2}<p$ for $\rho<1$. The weakness of this bound in our setup is that the function $q(\rho,p)$ does not depend on the support of $f$, which is exponentially small. It is quite plausible that deriving support dependent reverse hypercontractivity bounds, analogous to the support dependent hypercontractivity bounds of~\cite{ps16}, would result in tighter upper bounds on  $\underline{E}(\alpha,\beta,\rho)$ in the high-correlation regime.
\label{remark:morss}
\end{remark}


\section*{Acknowledgement}

The authors would like to thank the anonymous reviewers and the associate editor for their excellent suggestions, and in particular for a simplification of the proof of Proposition~\ref{prop:convexity}.

\begin{appendix}
\section{Proof of Proposition~\ref{prop:convexityScalar}}
Let $\phi(x)=\left(1-2 h^{-1}(x)\right)^2$. Since $h^{-1}(x)\left(1-h^{-1}(x)\right)=1-\frac{\phi(x)}{4}$, it suffices to show that $x\mapsto\phi(x)$ is concave. We have
\begin{align}
\phi'(x)&=-\frac{4}{\log\left(\frac{1-h^{-1}(x)}{h^{-1}(x)}\right)}\left(1-2h^{-1}(x)\right)\nonumber\\
&= -\frac{4}{\log{e}}v(h^{-1}(x)),
\end{align}
where
\begin{align}
v(t)=\frac{1-2t}{\ln\left(\frac{1-t}{t}\right)}.
\end{align}
Showing that $x\mapsto\phi(x)$ is concave is equivalent to showing that $x\mapsto\phi'(x)$ is decreasing, which in turn is equivalent to showing that $t\mapsto v(t)$ is increasing in $(0,1/2)$, due to monotonicity of $x\mapsto h^{-1}(x)$. Thus, it remains to show that $v'(t)\geq 0$ for $t\in(0,1/2)$. Let $y=y_t=\frac{1-t}{t}\in(1,\infty)$. We have that $v'(t)=\frac{\frac{1-2t}{t(1-t)}-2\ln\left(\frac{1-t}{t}\right)}{\ln^2\left(\frac{1-t}{t}\right)}$ and since $\frac{1-2t}{t(1-t)}=\frac{y^2-1}{y}$, it suffices to show that $g(y)=\frac{y^2-1}{y}-2\ln(y)\geq 0$ for all $y>1$. Noting that $g(1)=0$ and $g'(y)=1+\frac{1}{y^2}-\frac{2}{y}=\frac{(y-1)^2}{y^2}\geq 0$ for all $y\geq 1$, we see that indeed $g(y)\geq 0$ for all $y\geq 1$, which establishes our claim.
\end{appendix}

\bibliographystyle{IEEEtran}
\bibliography{ReportBib}

\begin{IEEEbiographynophoto}{Or Ordentlich} is a senior lecturer (assistant professor) in the School of Computer Science and Engineering at the Hebrew University of Jerusalem. He received the B.Sc. (cum laude), M.Sc. (summa cum laude),
	and Ph.D. degrees from Tel Aviv University, Israel, in 2010, in 2011, and
	2016, respectively, all in electrical engineering. During the years 2015-2017 he was a postdoctoral fellow in the Laboratory for Information and Decision Systems at the Massachusetts Institute of Technology (MIT), and in the Department of Electrical and Computer Engineering at Boston University. 
\end{IEEEbiographynophoto}
\begin{IEEEbiographynophoto}{Yury Polyanskiy}
 is an
	Associate Professor of Electrical Engineering and Computer Science and a member of IDSS and LIDS at MIT.
	Yury received M.S. degree in applied mathematics and physics from the Moscow Institute of Physics and Technology,
	Moscow, Russia in 2005 and Ph.D. degree in electrical engineering from Princeton
	University, Princeton, NJ in 2010. His research interests span information theory, statistical learning, error-correcting codes, wireless communication and fault tolerance.
	Dr. Polyanskiy won the 2020 IEEE Information Theory Society James Massey Award, 2013 NSF CAREER award and 2011 IEEE Information Theory Society Paper Award.
\end{IEEEbiographynophoto}
\begin{IEEEbiographynophoto}{Ofer Shayevitz} received the B.Sc. degree from the Technion Institute of Technology, Haifa, Israel, in 1997 and the M.Sc. and Ph.D. degrees from the Tel-Aviv University, Tel Aviv, Israel, in 2004 and 2009, respectively, all in electrical engineering. He is currently an Associate Professor in the Department of EE - Systems at Tel Aviv University, and serves as the head of the Advanced Communication Center (ACC). Ofer's research spans a wide cross-section of problems in information theory, statistical signal processing, and discrete mathematics. He is the recipient of the European Research Council (ERC) Starting Grant (2015--2020), and his group's research activities have further been supported by Israel Science Foundation grants (2014--2022), and the Marie Curie Grant (2014--2018). Ofer is also actively involved in the Israeli hi-tech industry, and regularly consults to various startup companies. Before joining Tel Aviv University, Ofer was a postdoctoral fellow in the Information Theory and Applications (ITA) Center at the University of California, San Diego (2008 - 2011), and worked as a quantitative analyst with the D.E. Shaw group in New York (2011 - 2013). Prior to his graduate studies, he served as a digital communication engineer and team leader in the Israeli Defense Forces (1997 - 2003), and worked in statistical signal processing at CellGuide (2003 - 2004). 
\end{IEEEbiographynophoto}

\end{document}